\documentclass[10pt,twocolumn,confmode]{IEEEtran}
\usepackage{epsf}
\usepackage{epsfig}
\usepackage{times}
\usepackage{amssymb}
\usepackage{amsmath}
\usepackage{amsfonts}
\usepackage{latexsym}
\usepackage{cite}
\usepackage{url}

\newtheorem{theorem}{Theorem}

\newtheorem{lemma}{Lemma}

\newtheorem{proposition}{Proposition}

\newtheorem{remark}{Remark}

\newcommand{\BEAS}{\begin{eqnarray*}}
\newcommand{\EEAS}{\end{eqnarray*}}

\newcommand{\mb}[1]{\mathbf{#1}}
\newcommand{\trace}[1]{\text{Tr}\left({#1}\right)}

\begin{document}

\title{The Weighted Sum Rate Maximization in MIMO Interference Networks: The Minimax Lagrangian Duality and Algorithm}

\author{Lijun Chen\ \ \ \ \ \ \ \ \quad \quad \quad \quad \quad \quad Seungil You
\thanks{L. Chen is with the College of Engineering and Applied Science, University of Colorado, Boulder, CO 80309, USA (email: lijun.chen@colorado.edu). S. You is with the Division of Engineering and Applied Science, California Institute of Technology, Pasadena, CA 91125, USA (email: syou@caltech.edu).}
}

\maketitle

\begin{abstract}
We take a new perspective on the weighted sum-rate maximization in multiple-input multiple-output (MIMO) interference networks, by formulating an equivalent max-min problem. This seemingly trivial reformulation has significant implications: the Lagrangian duality of the equivalent max-min problem provides an elegant way to establish the sum-rate duality between an interference network and its reciprocal when such a duality exists, and more importantly, suggests a novel iterative minimax algorithm for the weighted sum-rate maximization. Moreover, the design and convergence proof of the algorithm use only general convex analysis. They apply and extend to any max-min problems with similar structure, and thus provide a general class of algorithms for such optimization problems. This paper presents a promising step and lends hope for establishing a general framework based on the minimax Lagrangian duality for characterizing the weighted sum-rate and developing efficient algorithms for general MIMO interference networks. 

\end{abstract}


\section{Introduction}\label{sect:intro}

The weighted sum-rate maximization is a fundamental problem in information theory and communications, and serves as a basis for many resource management and network design problems. It has a long history, with a rich literature from the classical water-filling structure for parallel Gaussian channels to more recent polite water filling algorithm \cite{liu2010duality} and iterative weighted MMSE algorithm \cite{shi2011iteratively} for MIMO interference channels, to just name a few. The weighted sum-rate maximization is in general a highly nonconvex and NP hard problem, and despite its importance and long history, remains open for general channels/networks. 

In this paper, we consider the weighted sum-rate maximization in a general MIMO interference network that consists of a set of interfering data links, each of them equipped with multiple antennas at the transmitter and receiver. The MIMO interference network, under many different names such as MIMO B-MAC and MIMO IBC, includes broadcast channels, multiple access channel, interference channels, small cell networks, and many other practical wireless networks as special cases. Specifically, we study the weighted sum-rate maximization with general linear constraints, assuming Gaussian transmit signal, Gaussian noise, and the availability of channel state information at the transmitter (Section \ref{sect:sm}). It  typifies a class of problems that are key to the next generation wireless communication networks where the interference is a limiting factor; and various algorithms have been proposed for this problem or its special cases, see, e.g., \cite{Caire_09ISIT_BC_linear_constraints,Caire_09sTSP_BC_intercell_interference,Jindal_IT05_IFBC,Giannakis_TIT11_IFPMIMO,Berry_MonoIFCpricing_ISIT09,Yu_IT04_MIMO_MAC_waterfilling_alg,Weiyu_IT06_DualIWF,Wei_07ITW_MultiuserWF,Zhang_09ISIT_BC_MAC_duality_linear_constraints,ZhangLan_09TWC_Weighted_rate_BC,ng2010linear,liu2010duality,shi2011iteratively,liu2012polite,Li-2014-Infocom}. 

In a recent work \cite{Li-2014-Infocom}, we and our coauthors propose a new algorithm for the weighted sum-rate maximization in the  MIMO interference network with the total power constraint and establish its convergence. The convergence proof suggests certain general structure behind the problem and its possible connection to the Lagrangian duality. In this paper, we identify such a structure and establish its connection to the minimax Lagrangian duality for the weighted sum-rate maximization with general linear constraints, and explore its implications for the sum-rate characterization and algorithm design.  Specifically, we reformulate the weighted sum-rate maximization as an equivalent max-min problem, by treating the interference-plus-noise covariance matrix definition as a constraint. This seemingly trivial reformulation has significant implications: the Lagrangian duality of the equivalent max-min problem provides an elegant way to establish the sum-rate duality between an interference network and its reciprocal when such a duality exists (Section \ref{sect:md}), and more importantly, suggests a novel algorithm, termed {\em the iterative minimax algorithm}, for the weighted sum-rate maximization (Section \ref{sect:ima}). Moreover, the design and convergence proof of the algorithm use only general convex analysis. They apply and extend to any max-min problems where the objective function is concave in the maximizing variables and convex in the minimizing variables and the constraints are convex, and thus provides a general class of algorithms for such optimization problems.

The iterative minimax algorithm is based partially on an explicit saddle point solution of certain max-min optimization (Section \ref{sect:mdd}). This explicit solution has been identified for the case where the matrices involved are all invertible in~\cite{Yu-2006-IT}. In contrast, we prove the explicit solution for any general matrices, as long as the objective function is well-defined in a proper sense (the Appendix). Our proof uses only general matrix analysis, and the construction and techniques used are expected to find applications in handling singularity issues that arise from the matrix form capacity formula. 

This paper benefits from the insight in and to some extent can be seen as a substantial extension of the seminal work by Yu \cite{Yu-2006-IT} that establishes uplink-downlink duality via minimax duality for the sum capacity of the Gaussian broadcast channel. Our model is much more general and the results expect to find broad applications, and we establish the explicit saddle point solution for the max-min optimization with general matrices, and more importantly, we develop a novel algorithm for the weighted sum-rate maximization.

%\lijun{Add an explanation of notation...}

\section{System Model}\label{sect:sm}

Consider a general interference network with a set $L$ of MIMO links or users, with the transmitter $t_l$ and receiver $r_l$ of link $l\in L$ being equipped with $n_l$  and $m_l$ antennas respectively. Let $\mathbf{x}_l\in\mathbb{C}^{n_l\times 1}$ denote the transmit signal of link $l$, which is assumed to be circularly symmetric complex Gaussian. The received signal $\mathbf{y}_l\in \mathbb{C}^{m_l\times 1}$ at the receiver $r_l$ can be written as
\begin{eqnarray}
\mathbf{y}_l=\sum_{k\in L}\mathbf{H}_{lk}\mathbf{x}_k +\mathbf{w}_l, \label{eq:rs} 
\end{eqnarray}
where $\mathbf{H}_{lk}\in \mathbb{C}^{m_l\times n_l}$ denotes the channel matrix from the transmitter $t_k$ to the receiver $r_l$, and $\mathbf{w}_l\in \mathbb{C}^{m_l\times 1}$ denotes the additive circularly symmetric complex Gaussian noise with identity covariance matrix. 

The interference network defined above is very general and includes as special cases many practical channels and networks such as broadcast channels, multiple access channels, small cell networks, and heterogeneous networks, etc.

\subsection{The power covariance constraints}
Denote by $\mathbf{\Sigma}_l\succeq 0$ the covariance matrix of the transmit signal $\mathbf{x}_l,~l\in L$. We now specify the constraints on these power covariance matrices.

Assume that the links are grouped into a set $S$ of non-empty subsets $L^s,~s\in S$ that cover all of $L$. Each subset $L^s$ may correspond to those links that are controlled or managed by a certain entity or for a certain purpose. These subsets may overlap with each other. For each link $l\in L$, denote by $S^l$ the set of those subsets that include the link, i.e., $S^l=\{s\in S| l\in L^s\}$.

Each link $l\in L$ is associated with an $n_l\times n_l$ constraint matrix $\mathbf{Q}_l^s\succ 0$ for each $s\in S^l$; and any two of these matrices may be identical. We assume that each group of links $L^s,~s\in S$ is subject to a linear power covariance constraint as follows: 
\begin{eqnarray}
\sum_{l\in L^s}\trace{\mathbf{\Sigma}_l \mathbf{Q}_l^s}\leq 1,~s\in S.\label{eq:pc}
\end{eqnarray}

The constraint (\ref{eq:pc}) is very general and captures all reasonble linear power constraints. For example, when there is only a budget $P_T$ on the total power of all links as considered in many existing work such as \cite{Li-2014-Infocom}, the cardinality $|S|=1$ and $\mathbf{Q}_l^s=\frac{1}{P_T}\mathbf{I}$. When there is only a per-link power budget $p_l,~l\in L$, each group $L^s$ contains only one link and $\mathbf{Q}_l^s=\frac{1}{P_l}\mathbf{I}$. Each group $L^s, s\in S$ may also represent those links or users in a cell of a microcell network and each cell $s$ is subject to a total power budget $P_s$. In this scenario, the subsets $L^s$ are non-overlapping and $\mathbf{Q}_l^s=\frac{1}{P_s}\mathbf{I},~\forall l\in L^s$. 

\remark{We have assumed linear power covariance constraints. However, as will be seen later, our theory development and algorithm design are based on general convex analysis, so the results in this paper can be extended to the network with nonlinear convex power covariance constraints.}

\subsection{The weighted sum-rate maximization}\label{subsec:wsrm}

Assume that the channel state information is known. For given power covariance matrix $\mathbf{\Sigma}_l,~l\in L$, an achievable rate $R_l$ of the link $l$ is given by
\begin{eqnarray}
R_l=\log{\left|\mathbf{I} + \mathbf{H}_{ll}\mathbf{\Sigma}_l \mathbf{H}_{ll}^+ \left(\mathbf{I}+\sum_{k\in L\backslash \{l\}} \mathbf{H}_{lk}\mathbf{\Sigma}_k \mathbf{H}_{lk}^+\right)^{-1}\right|}, \label{eq:ir}
\end{eqnarray}
where $|\cdot|$ denotes the matrix determinant and the interferences from other links are treated as noise. Assume that each link $l\in L$ is associated with a weight $w_l>0$. We aim to allocate power for each link so as to maximize the weighted sum-rate subject to the power constraints:
\begin{eqnarray}
\max_{\mathbf{\Sigma}_l\succeq 0}&&\sum_{l\in L} w_l R_l \label{eq:srmo}\\
\mbox{s.t.} && \sum_{l\in L^s}\trace{\mathbf{\Sigma}_l \mathbf{Q}_l^s}\leq 1,~s\in S. \label{eq:srmc}
\end{eqnarray}
The weighted sum-rate maximization is in general a hard nonconvex problem. It is a fundamental problem in information theory and communications and serves as a basis for many resource management and network design problems, while still remains open for general channels/networks.

\section{The minimax Lagrangian duality}\label{sect:md}
In this section, we will reformulate the weighted sum-rate maximization as an equivalent max-min problem, by treating the interference-plus-noise covariance matrix definition as a constraint. This seemingly trivial reformulation has significant implications: the Lagrangian duality of the equivalent max-min problem provides an elegant way to establish the sum-rate duality between an interference network and its reciprocal when such a duality exists, and more importantly, suggests a new algorithm for the weighted sum-rate maximization. 

\subsection{The minimax Lagrangian duality}\label{sect:mdd}
Denote by $\mathbf{\Omega}_l,~l\in L$ the interference-plus-noise covariance matrix at the receiver $r_l$, i.e.,
\begin{eqnarray}
\mathbf{\Omega}_l=\mathbf{I}+\sum_{k\in L\backslash \{l\}} \mathbf{H}_{lk}\mathbf{\Sigma}_k \mathbf{H}_{lk}^+.\label{eq:ipn}
\end{eqnarray}
We can rewrite the weighted sum-rate maximization (\ref{eq:srmo})-(\ref{eq:srmc}) equivalently as the following max-min problem:
\begin{eqnarray}
\max_{\mathbf{\Sigma}_l\succeq 0}\min_{\mathbf{\Omega}_l\succeq 0}&&\hspace{-3mm} \sum_{l\in L} w_l \left(\log{\left|\mathbf{\Omega}_l+ \mathbf{H}_{ll}\mathbf{\Sigma}_l \mathbf{H}_{ll}^+\right|}- \log{\left|\mathbf{\Omega}_l\right| }\right) \label{eq:maxmin-ob}\\
\mbox{s.t.}&&\hspace{-3mm}\sum_{l\in L^s}\trace{\mathbf{\Sigma}_l\mathbf{Q}_l^s}\leq 1,~s\in S, \label{eq:maxmin-con1}\\
&& \hspace{-3mm}\mathbf{\Omega}_l=\mathbf{I}+\sum_{k\in L\backslash \{l\}} \mathbf{H}_{lk}\mathbf{\Sigma}_k \mathbf{H}_{lk}^+,~l\in L.\label{eq:maxmin-con2}
\end{eqnarray}
Note that, when $\mathbf{H}_{ll}\mathbf{\Sigma}_l \mathbf{H}_{ll}^+$ is not of full rank, the above problem is equivalent to a {\em truncated system} where $\mb{\Omega}_l$ is restricted to $\mb{\Omega}_l=\mathbf{H}_{ll}\mathbf{X}_l \mathbf{H}_{ll}^+,~\mathbf{X}_l\succeq 0$. Intuitively, this follows from the fact that when the signal at a channel is zero, it does not matter what the interference-plus-noise is, in terms of the achieved rate; mathematically, this causes technical difficulty regarding singular matrices; see the Appendix for more detail and insight. 

The objective function of problem \eqref{eq:maxmin-ob}-\eqref{eq:maxmin-con2}
\begin{eqnarray*}
\mathcal{F}(\mathbf{\Sigma}, \mathbf{\Omega}) = \sum_{l\in L} w_l \left( \log{\left|\mathbf{\Omega}_l+ \mathbf{H}_{ll}\mathbf{\Sigma}_l \mathbf{H}_{ll}^+\right|}- \log{\left|\mathbf{\Omega}_l\right|} \right) 
\end{eqnarray*}
is concave in $\mathbf{\Sigma}$ and convex in $\mathbf{\Omega}$. So, the max-min is equal to min-max, and the optimum is a saddle point. Consider the Lagrangian 
\begin{eqnarray*}
&&\mathcal{L}(\mathbf{\Sigma}, \mathbf{\Omega}, \mathbf{\Lambda}, \mathbf{\mu})\\
&=& \mathcal{F}(\mathbf{\Sigma}, \mathbf{\Omega}) +\sum_{s\in S}\mu_s\left(1- \sum_{l\in L^s} \trace{\mathbf{\Sigma}_l\mathbf{Q}_l^s}\right)\\
&& +\sum_{l\in L}\trace{\mathbf{\Lambda}_l (\mathbf{\Omega}_l - \mathbf{I} - \sum_{k\in L\backslash \{l\}} \mathbf{H}_{lk}\mathbf{\Sigma}_{k}\mathbf{H}_{lk}^+)},
\end{eqnarray*}
where $\mu_s \geq 0$ is the dual variable associated with the power constraint (\ref{eq:maxmin-con1}) and $\mathbf{\Lambda}_l\succeq 0$ is the dual variable associated with the interference-plus-noise covariance matrix definition (\ref{eq:maxmin-con2}).\footnote{Even though equation \eqref{eq:maxmin-con2} is an equality constraint, the dual feasibility requires $\mathbf{\Lambda}_l\succeq 0$.} For any given $(\mathbf{\Lambda}, \mathbf{\mu})$, $\mathcal{L}$ is concave in $\mathbf{\Sigma}$ and convex in $\mathbf{\Omega}$ as $\mathcal{F}$ is.  

Consider  the first order condition (part of the KKT condition \cite{Boyd}) for the optimum:\footnote{Note that the first oder condition does not hold for all dual variables, but only for those that satisfy the dual feasibility condition. We only need to consider those feasible dual variables~\cite{Boyd}.}
\begin{eqnarray}
w_l \mathbf{H}_{ll}^{+} \left(\mathbf{\Omega}_l + \mathbf{H}_{ll}\mathbf{\Sigma}_{l}\mathbf{H}_{ll}^{+}\right)^{-1} \mathbf{H}_{ll}=\mb{\Phi}_l,\label{eq:kkt1}\\
w_l \left(\mathbf{\Omega}_l^{-1}-\left(\mathbf{\Omega}_l + \mathbf{H}_{ll}\mathbf{\Sigma}_{l}\mathbf{H}_{ll}^{+}\right)^{-1}\right)  =\mathbf{\Lambda}_l,\label{eq:kkt2}
\end{eqnarray}
where
\begin{eqnarray*}
\mb{\Phi}_l = \sum_{s\in S_l}\mu_s \mathbf{Q}_l^s + \sum_{k \in L\backslash \{l\}} \mathbf{H}_{kl}^+ \mathbf{\Lambda}_k \mathbf{H}_{kl}.
\end{eqnarray*}
For any given feasible dual variable $(\mathbf{\Lambda}, \mathbf{\mu})$, the above condition gives the saddle point condition of Lagrangian $\mathcal{L}$ as a function of $(\mathbf{\Sigma}, \mathbf{\Omega})$; and when $(\mathbf{\Lambda}, \mathbf{\mu})$ is a dual optimum, solving (\ref{eq:kkt1})-(\ref{eq:kkt2}) gives a primal optimum \cite{Boyd}. In the next section we will exploit this fact to design a novel algorithm to solve the weighted sum-rate maximization. 

\begin{lemma}\label{thm:es}
Given feasible dual variables $( \mb{\Phi}, \mb{\mu})$, an explicit solution $(\mb{\Sigma}, \mb{\Omega})$ for the saddle point equations (\ref{eq:kkt1})-(\ref{eq:kkt2}) is given by:\footnote{Note that at an optimum, with general channel matrix, equation \eqref{eq:es1} may only give a solution of the equivalent truncated system but not the solution of the original max-min problem. In order to obtain a solution for $\mb{\Omega}$ of the original problem, we should use a generalized solution $\mathbf{\Omega}_{l}= w_l \mathbf{H}_{ll} \left(\mb{\Phi}_l + \mathbf{H}_{ll}^+\mb{\Lambda}_l\mathbf{H}_{ll}\right)^{-1} \mathbf{H}_{ll}^+ + \mathbf{\Omega}_{l}^c$. Here $\mathbf{\Omega}_{l}^c$ should satisfy certain proper condition, but  at an optimum it can be easily determined according to equation \eqref{eq:maxmin-con2}. We will not elaborate on this mathematical peculiarity, as it does not affect the results presented in this paper.}
\begin{eqnarray}
w_l \mathbf{H}_{ll} \left(\mb{\Phi}_l + \mathbf{H}_{ll}^+\mb{\Lambda}_l\mathbf{H}_{ll}\right)^{-1} \mathbf{H}_{ll}^+ = \mathbf{\Omega}_{l},\label{eq:es1}\\
w_l \left(\mb{\Phi}_l^{-1}-\left(\mb{\Phi}_l + \mathbf{H}_{ll}^+ \mb{\Lambda}_l \mathbf{H}_{ll}\right)^{-1}\right)  =\mathbf{\Sigma}_l.\label{eq:es2}
\end{eqnarray}
\end{lemma}

\vspace{3mm}
The solution (\ref{eq:es1})-(\ref{eq:es2}) is motivated by \cite{Yu-2006-IT} that focuses on an (primal-dual) optimum and where correspondingly the optimal power covariance matrix $\mb{\Sigma}_l$ and the interference-plus-noise matrix $\mb{\Omega}_l$ are assumed to be positive definite and the channel matrix $\mb{H}_{ll}$ is assumed to be square and invertible. Here, the solution is for any given feasible dual variables, and the power covariance matrix and the interference-plus-noise matrix are positive semidefinite and the channel matrix can be any general matrix. However, the solution is for an equivalent, {\em truncated system} where we ignore the interference-plus-noise of a channel whose signal is zero, and ``-1'' denotes pseudo inverse if the matrix involved is singular. The proof of Lemma \ref{thm:es} is rather involved, and is presented in the Appendix.

Equations (\ref{eq:kkt1})-(\ref{eq:kkt2}) and equations (\ref{eq:es1})-(\ref{eq:es2}) have similar structures, which can be exploited to establish the sum-rate duality between an interference network and its reciprocal based on the Lagrangian dual of the (truncated) max-min problem \eqref{eq:maxmin-ob}-\eqref{eq:maxmin-con2}.

\subsection{Case studies}
We now discuss two typical cases, and show how the minimax Lagrangian duality can be used to establish the rate duality between the interference network and its reciprocal. 

\subsubsection{The network with the per-link power constraints and without interlink interference} 
Here the set $S=L$,  and $\mathbf{\Omega}_l =\mathbf{I}$ and $\mathbf{Q}_l=\frac{\mathbf{I}}{P_l}$, with $P_l$ the power budget at each link $l\in L$.  As each link is independent, we can just focus on one link:
\begin{eqnarray}
\max_{\mathbf{\Sigma}_l\succeq 0}\min_{\mathbf{\Omega}_l\succeq 0}&&\hspace{-3mm}  \mbox{log}|\mathbf{\Omega}_l+ \mathbf{H}_{ll}\mathbf{\Sigma}_l \mathbf{H}_{ll}^+|- \mbox{log}|\mathbf{\Omega}_l| \label{eq:maxmin-obi}\\
\mbox{s.t.}&&\hspace{-3mm}\trace{\frac{\mathbf{\Sigma}_l}{P_l}} \leq 1,~~\mathbf{\Omega}_l=\mathbf{I}.\label{eq:maxmin-con1i}
\end{eqnarray}
The first order condition (\ref{eq:kkt1})-(\ref{eq:kkt2}) reduces to
\begin{eqnarray}
\nonumber w_l \mathbf{H}_{ll}^{+} \left(\mathbf{\Omega}_l + \mathbf{H}_{ll}\mathbf{\Sigma}_{l}\mathbf{H}_{ll}^{+}\right)^{-1} \mathbf{H}_{ll}= \mu_l \frac{\mathbf{I}}{P_l},\label{eq:kkt1i}\\
\nonumber w_l \left(\mathbf{\Omega}_l^{-1}-\left(\mathbf{\Omega}_l + \mathbf{H}_{ll}\mathbf{\Sigma}_{l}\mathbf{H}_{ll}^{+}\right)^{-1}\right)  =\mathbf{\Lambda}_l,\label{eq:kkt2i}
\end{eqnarray}
where $\mu_l\geq 0$ is the dual variable associated with the power constraint. Define
\begin{eqnarray}
\nonumber \hat{\mathbf{\Sigma}}_l&=& \frac{P_l}{\mu_l} \mathbf{\Lambda}_l,\label{eq:dr-1i}\\
\nonumber \hat{\mathbf{\Omega}}_l&=& \mathbf{I}.\label{eq:dr-2i}
\end{eqnarray}
The first order condition becomes
\begin{eqnarray}
w_l \mathbf{H}_{ll}^{+} \left(\mathbf{\Omega}_l + \mathbf{H}_{ll}\mathbf{\Sigma}_{l}\mathbf{H}_{ll}^{+}\right)^{-1} \mathbf{H}_{ll}=\frac{\mu_l}{P_l} \hat{\mathbf{\Omega}}_{l},\label{eq:kkt1ei}\\
w_l \left(\mathbf{\Omega}_l^{-1}-\left(\mathbf{\Omega}_l + \mathbf{H}_{ll}\mathbf{\Sigma}_{l}\mathbf{H}_{ll}^{+}\right)^{-1}\right)  =\frac{\mu_l}{P_l}\hat{\mathbf{\Sigma}}_l,\label{eq:kkt2ei}
\end{eqnarray}
and the explicit solution (\ref{eq:es1})-(\ref{eq:es2}) becomes  
\begin{eqnarray}
w_l \mathbf{H}_{ll} \left(\hat{\mathbf{\Omega}}_l + \mathbf{H}_{ll}^+\hat{\mathbf{\Sigma}}_{l}\mathbf{H}_{ll}\right)^{-1} \mathbf{H}_{ll}^+ = \frac{\mu_l}{P_l} \mathbf{\Omega}_{l},\label{eq:es1ti}\\
w_l \left(\hat{\mathbf{\Omega}}_l^{-1}-\left(\hat{\mathbf{\Omega}}_l + \mathbf{H}_{ll}^+\hat{\mathbf{\Sigma}}_{l}\mathbf{H}_{ll}\right)^{-1}\right)  =\frac{\mu_l}{P_l} \mathbf{\Sigma}_l.\label{eq:es2ti}
\end{eqnarray}

Compare equations (\ref{eq:kkt1ei})-(\ref{eq:kkt2ei}) and equations (\ref{eq:es1ti})-(\ref{eq:es2ti}), we can conclude that the Lagrangian dual of the max-min problem (\ref{eq:maxmin-obi})-(\ref{eq:maxmin-con1i}) is also a max-min problem:
\begin{eqnarray}
\max_{\hat{\mathbf{\Sigma}}_l\succeq 0}\min_{\hat{\mathbf{\Omega}}_l\succeq 0}&&\hspace{-3mm}  \mbox{log}|\hat{\mathbf{\Omega}}_l+ \mathbf{H}_{ll}^+\hat{\mathbf{\Sigma}}_l \mathbf{H}_{ll}|- \mbox{log}|\hat{\mathbf{\Omega}}_l| \label{eq:dmaxmin-obi}\\
\mbox{s.t.}&&\hspace{-3mm}\trace{\frac{\hat{\mathbf{\Sigma}_l}}{P_l}} \leq 1,~~\hat{\mathbf{\Omega}}_l=\mathbf{I},\label{eq:dmaxmin-con1i}
\end{eqnarray}
which is the sum-rate maximization problem defined on the reciprocal link with channel matrix $\mathbf{H}^+_l$. At the corresponding saddle points, the two problems achieve the same rate, since one is the dual of the other. Furthermore, introducing the dual variables $\hat{\mu}$ and $\hat{\mathbf{\Lambda}}_l$  for the problem (\ref{eq:dmaxmin-obi})-(\ref{eq:dmaxmin-con1i}), we have the following correspondence:
\begin{eqnarray}
(\mathbf{\Sigma}_l; \mathbf{\Lambda}_l, \mu_l) = (\frac{P_l}{\hat{\mu}_l}\hat{\mathbf{\Lambda}}_l ; \frac{\hat{\mu}_l}{P_l}\hat{\mathbf{\Sigma}}_l, \hat{\mu}_l),\label{eq:corres1i} \\
(\hat{\mathbf{\Sigma}}_l; \hat{\mathbf{\Lambda}}_l, \hat{\mu}_l) = (\frac{P_l}{\mu_l} \mathbf{\Lambda}_l; \frac{\mu_l}{P_l} \mathbf{\Sigma}_l, \mu_l).\label{eq:corres2i}
\end{eqnarray}

This recovers the well-known result in \cite{Viswanath-2003-IT,VJG-2003-IT,Yu-2006-IT}. The difference from \cite{Yu-2006-IT} is that we establish the explicit solution (\ref{eq:es1ti})-(\ref{eq:es2ti}) and the correspondence (\ref{eq:corres1i})-(\ref{eq:corres2i}) for general power covariance matrices and channel matrices and at any saddle points of the Lagrangian function (instead of only at an optimum). 

\subsubsection{The network with the total power constraint}\label{sect:md-t}
Here $|S|=1$ and $\mathbf{Q}_l=\frac{\mathbf{I}}{P_T}$, with $P_T$ the total power budget. The max-min problem (\ref{eq:maxmin-ob})-(\ref{eq:maxmin-con2}) reduces to 
\begin{eqnarray}
\max_{\mathbf{\Sigma}_l\succeq 0}\min_{\mathbf{\Omega}_l\succeq 0}&&\hspace{-3mm} \sum_l w_l \left( \mbox{log}|\mathbf{\Omega}_l+ \mathbf{H}_{ll}\mathbf{\Sigma}_l \mathbf{H}_{ll}^+|- \mbox{log}|\mathbf{\Omega}_l| \right) \label{eq:maxmin-obt}\\
\mbox{s.t.}&&\hspace{-3mm}\sum_l \trace{\frac{\mathbf{\Sigma}_l}{P_T}} \leq 1,\label{eq:maxmin-con1t}\\
&&\hspace{-3mm} \mathbf{\Omega}_l=\mathbf{I}+\sum_{k\in L\backslash \{l\}} \mathbf{H}_{kl}\mathbf{\Sigma}_k \mathbf{H}_{lk}^+,\label{eq:maxmin-con2t}
\end{eqnarray}
and the first order condition (\ref{eq:kkt1})-(\ref{eq:kkt2}) reduces to
\begin{eqnarray}
\nonumber w_l \mathbf{H}_{ll}^{+} \left(\mathbf{\Omega}_l + \mathbf{H}_{ll}\mathbf{\Sigma}_{l}\mathbf{H}_{ll}^{+}\right)^{-1} \mathbf{H}_{ll}= \mathbf{\Phi}_l,\label{eq:kkt1t}\\
\nonumber w_l \left(\mathbf{\Omega}_l^{-1}-\left(\mathbf{\Omega}_l + \mathbf{H}_{ll}\mathbf{\Sigma}_{l}\mathbf{H}_{ll}^{+}\right)^{-1}\right)  =\mathbf{\Lambda}_l,\label{eq:kkt2t}
\end{eqnarray}
with
$\mathbf{\Phi}_l = \mu \frac{\mathbf{I}}{P_l} + \hspace{-2mm}\sum_{k\in L\backslash \{l\}} \mathbf{H}_{kl}^+ \mathbf{\Lambda}_k \mathbf{H}_{kl}$, 
where $\mu\geq 0$ is the dual variable associated with the total power constraint. Define
\begin{eqnarray}
\nonumber \hat{\mathbf{\Sigma}}_l&=& \frac{P_T}{\mu} \mathbf{\Lambda}_l,\label{eq:dr-1}\\
\nonumber \hat{\mathbf{\Omega}}_l&=& \mathbf{I}+\sum_{k\in L\backslash \{l\}} \mathbf{H}_{kl}^+\hat{\mathbf{\Sigma}}_k \mathbf{H}_{kl}.\label{eq:dr-2}
\end{eqnarray}
The first order condition becomes
\begin{eqnarray}
w_l \mathbf{H}_{ll}^{+} \left(\mathbf{\Omega}_l + \mathbf{H}_{ll}\mathbf{\Sigma}_{l}\mathbf{H}_{ll}^{+}\right)^{-1} \mathbf{H}_{ll}=\frac{\mu}{P_T} \hat{\mathbf{\Omega}}_{l},\label{eq:kkt1e}\\
w_l \left(\mathbf{\Omega}_l^{-1}-\left(\mathbf{\Omega}_l + \mathbf{H}_{ll}\mathbf{\Sigma}_{l}\mathbf{H}_{ll}^{+}\right)^{-1}\right)  =\frac{\mu}{P_T}\hat{\mathbf{\Sigma}}_l,\label{eq:kkt2e}
\end{eqnarray}
and the explicit solution (\ref{eq:es1})-(\ref{eq:es2}) becomes  
\begin{eqnarray}
w_l \mathbf{H}_{ll} \left(\hat{\mathbf{\Omega}}_l + \mathbf{H}_{ll}^+\hat{\mathbf{\Sigma}}_{l}\mathbf{H}_{ll}\right)^{-1} \mathbf{H}_{ll}^+ = \frac{\mu}{P_T} \mathbf{\Omega}_{l},\label{eq:es1t}\\
w_l \left(\hat{\mathbf{\Omega}}_l^{-1}-\left(\hat{\mathbf{\Omega}}_l + \mathbf{H}_{ll}^+\hat{\mathbf{\Sigma}}_{l}\mathbf{H}_{ll}\right)^{-1}\right)  =\frac{\mu}{P_T} \mathbf{\Sigma}_l.\label{eq:es2t}
\end{eqnarray}

Compare equations (\ref{eq:kkt1e})-(\ref{eq:kkt2e}) and equations (\ref{eq:es1t})-(\ref{eq:es2t}), we can conclude that the Lagrangian dual of the max-min problem (\ref{eq:maxmin-obt})-(\ref{eq:maxmin-con2t}) is also a max-min problem:
\begin{eqnarray}
\max_{\hat{\mathbf{\Sigma}}_l\succeq 0}\min_{\hat{\mathbf{\Omega}}_l\succeq 0}&&\hspace{-3mm} \sum_l w_l \left( \mbox{log}|\hat{\mathbf{\Omega}}_l+ \mathbf{H}_{ll}^+\hat{\mathbf{\Sigma}}_l \mathbf{H}_{ll}|- \mbox{log}|\hat{\mathbf{\Omega}}_l| \right) \label{eq:dmaxmin-obt}\\
\mbox{s.t.}&&\hspace{-3mm}\sum_l \trace{\frac{\hat{\mathbf{\Sigma}}_l}{P_T}} \leq 1,\label{eq:dmaxmin-con1t}\\
&&\hspace{-3mm} \hat{\mathbf{\Omega}}_l=\mathbf{I}+\sum_{k\in L\backslash \{l\}} \mathbf{H}_{kl}^+\hat{\mathbf{\Sigma}}_k \mathbf{H}_{lk},\label{eq:dmaxmin-con2t}
\end{eqnarray}
which is the weighted sum-rate maximization problem defined on a network of reciprocal channels with channel matrix $\mathbf{H}^+$. At the corresponding saddle points, the two problems achieve the same weighted sum-rate, since one is the dual of the other. Furthermore, introducing the dual variables $\hat{\mu}$ and $\hat{\mathbf{\Lambda}}_l$  for the problem (\ref{eq:dmaxmin-obt})-(\ref{eq:dmaxmin-con2t}), we have the following correspondence:
\begin{eqnarray}
(\mathbf{\Sigma}_l; \mathbf{\Lambda}_l, \mu) = (\frac{P_T}{\hat{\mu}}\hat{\mathbf{\Lambda}}_l ; \frac{\hat{\mu}}{P_T}\hat{\mathbf{\Sigma}}_l, \hat{\mu}), \\
(\hat{\mathbf{\Sigma}}_l; \hat{\mathbf{\Lambda}}_l, \hat{\mu}) = (\frac{P_T}{\mu} \mathbf{\Lambda}_l; \frac{\mu}{P_T} \mathbf{\Sigma}_l, \mu).
\end{eqnarray}
This provides a simple proof of the weighted sum-rate duality identified in, e.g., \cite{liu2010duality}.

%\lijun{Can dig deeper into the minimax duality, and should also present the case studies in a mathematical rigorous way. Leave these to the future...}

\section{The iterative minimax algorithm}\label{sect:ima}
Motivated by the minimax Lagrangian duality, in this section we will design a novel algorithm for the weighted sum-rate maximization and establish its convergence properties. Our algorithm applies/extends to any max-min problems where the objective function is concave in the maximizing variables and convex in the minimizing variables and the constraints are convex, and thus provides a general class of algorithms for such optimization problems.

\subsection{The iterative minimax algorithm}\label{sect:alg}

Note that the optimum of the max-min problem (\ref{eq:maxmin-ob})-(\ref{eq:maxmin-con2}) is a saddle point, and the first order condition (\ref{eq:kkt1})-(\ref{eq:kkt2}) or part of it will give a saddle point, maximum or minimum of Lagrangian $\mathcal{L}$ when certain subset of its variables is fixed and given. This motivates an iterative minimax algorithm to achieve an optimum, as follows. 

\begin{enumerate}
\item Start with given $\mathbf{\Sigma}_l^n,~l\in L$ that is feasible, i.e., 
\begin{eqnarray*}
\sum_{l\in L^s}\trace{\mathbf{\Sigma}_l^n\mathbf{Q}_l^s}\leq 1,~s\in S,
\end{eqnarray*} 
and $\mathbf{\Omega}_l^n=\mathbf{I}+\sum_{k\in L\backslash \{l\}} \mathbf{H}_{lk}\mathbf{\Sigma}_k^n \mathbf{H}_{lk}^+,~l\in L$. By equation (\ref{eq:kkt2}) that gives the condition for minimizing $\mathcal{L}$ over $\mathbf{\Omega}_l$, we choose $\mathbf{\Lambda}_l^n\succeq 0$ such that 
\begin{eqnarray}
\mathbf{\Lambda}_l^n=w_l\left((\mathbf{\Omega}_l^n)^{-1}-\left(\mathbf{\Omega}_l^n + \mathbf{H}_{ll}\mathbf{\Sigma}_{l}^n\mathbf{H}_{ll}^{+}\right)^{-1}\right).\label{eq:alg1a}
\end{eqnarray}
Therefore, for any $\mathbf{\Omega}\succeq 0$, we have
\begin{eqnarray}
 \nonumber \mathcal{F}(\mathbf{\Sigma}^n, \mathbf{\Omega}^n)&\leq&\mathcal{L}(\mathbf{\Sigma}^n, \mathbf{\Omega}^n, \mathbf{\Lambda}^n, \mathbf{\mu}^n) \\
 &\leq& \mathcal{L} (\mathbf{\Sigma}^n, \mathbf{\Omega}, \mathbf{\Lambda}^n, \mathbf{\mu}^n), \label{eq:ie-min}
 \end{eqnarray}
 where $\mathbf{\mu}^n\succeq 0$ will be determined later. Define
 \begin{eqnarray}
 \mathbf{\Phi}_l^n=\sum_{s\in S^l}\mu_s^n \mathbf{Q}_l^s+\sum_{k\in L\backslash \{l\}} \mathbf{H}_{kl}^+\mathbf{\Lambda}_k^n \mathbf{H}_{kl}.\label{eq:alg1b}
 \end{eqnarray}
 Note that $ \mathbf{\Phi}_l^n$ does not necessary satisfy equation (\ref{eq:kkt1}).
 
 \item Given the above $(\mathbf{\Lambda}_l^n, \mathbf{\Phi}_l^n)$ and $\mathbf{\mu}^n$, by equations (\ref{eq:es1})-(\ref{eq:es2}), we choose $(\tilde{\mathbf{\Sigma}}^{n+1}, \tilde{\mathbf{\Omega}}^{n+1})$ such that 
 \begin{eqnarray}
\tilde{\mathbf{\Sigma}}_l^{n+1}=w_l\left((\mathbf{\Phi}_l^n)^{-1}-\left(\mathbf{\Phi}_l ^n+ \mathbf{H}_{ll}\mathbf{\Lambda}_{l}^n\mathbf{H}_{ll}^{+}\right)^{-1}\right). \label{eq:alg2a}\\
\tilde{\mathbf{\Omega}}_l^{n+1}=w_l \mb{H}_{ll}\left( \mathbf{\Phi}_l ^n+ \mathbf{H}_{ll}\mathbf{\Lambda}_{l}^n\mathbf{H}_{ll}^{+}\right)^{-1}\mb{H}_{ll}^+.\label{eq:alg2b}
 \end{eqnarray}
Plug $\mathbf{\Omega}=\tilde{\mathbf{\Omega}}^{n+1}$ into inequality (\ref{eq:ie-min}), we have
\begin{eqnarray}
\mathcal{F}(\mathbf{\Sigma}^n, \mathbf{\Omega}^n) \leq \mathcal{L}(\mathbf{\Sigma}^n, \tilde{\mathbf{\Omega}}^{n+1}, \mathbf{\Lambda}^n, \mb{\mu}^n) . \label{eq:ie-mins}
 \end{eqnarray}
By the first order condition (\ref{eq:kkt1})-(\ref{eq:kkt2}), $(\tilde{\mathbf{\Sigma}}^{n+1}, \tilde{\mathbf{\Omega}}^{n+1})$ is the saddle point of $\mathcal{L}(\mathbf{\Sigma}, \mathbf{\Omega}, \mathbf{\Lambda}^n, \mu^n)$. Thus,
\begin{eqnarray}
\nonumber \mathcal{L}(\mathbf{\Sigma}^n, \tilde{\mathbf{\Omega}}^{n+1}, \mathbf{\Lambda}^n, \mu^n)\hspace{-2mm}&\leq& \hspace{-2mm} \mathcal{L}(\tilde{\mathbf{\Sigma}}^{n+1}, \tilde{\mathbf{\Omega}}^{n+1}, \mathbf{\Lambda}^n, \mu^n) \\
\hspace{-2mm}&\leq&\hspace{-2mm} L(\tilde{\mathbf{\Sigma}}^{n+1}, \mathbf{\Omega}, \mathbf{\Lambda}^n, \mu^n)
\label{eq:ie-max}
\end{eqnarray}
for any $\mathbf{\Omega}\succeq 0$.
 
\item The matrix $\tilde{\mathbf{\Sigma}}^{n+1}_l$ is a function of $\mu_s^n,~s\in S^l$, denoted explicitly by $\tilde{\mathbf{\Sigma}}^{n+1}_l (\mu_s^n; s\in S^l)$. Define the set $T$ such that
 \begin{eqnarray*}
 T=\{s\in S| \sum_{l\in L^s}\trace{\mb{Q}_l^s\tilde{\mathbf{\Sigma}}^{n+1}_l (\mu_{\bar{s}}^n=0^+; \bar{s}\in S^l)}\geq 1\}.
 \end{eqnarray*}
For each $s\in S\backslash T$, we set $\mu_s^n=0$. For those $s\in T$, we choose $\mu_s^n$ such that  
\begin{eqnarray}
\sum_{l\in L^s}\trace{\mb{Q}_l^s\tilde{\mathbf{\Sigma}}^{n+1}_l (\mu_{\bar{s}}^n; \bar{s}\in S^l)}= 1,~s\in T.\label{eq:alg3a}
\end{eqnarray}
Note that $\trace{\tilde{\mathbf{\Sigma}}^{n+1}_l (\mu_s^n; s\in S^l)}$ is decreasing in $\mu_s^n$, and there are $|T|$ equations for $|T|$ variables. So, there exists a solution to equation (\ref{eq:alg3a}). With the afore choice of $\mu_s^n,~s\in S$, we can see that
\begin{eqnarray}
\mu_s^n \left(1-\sum_{l\in L^s}\trace{\tilde{\mathbf{\Sigma}}^{n+1}_l\mb{Q}_l^s}\right)=0.\label{eq:alg3b}
\end{eqnarray}

The above is a complementary slackness condition (part of the KKT condition) that is required at an optimum~\cite{Boyd}, but in our algorithm we enforce this condition at each iteration.
  
 \item Let 
 \begin{eqnarray*}
 \lambda=\max_{s\in S} \sum_{l\in L^s}\trace{\tilde{\mathbf{\Sigma}}^{n+1}_l \mb{Q}_l^s}.
 \end{eqnarray*}
We see that $0<\lambda\leq 1$. We then choose $(\mathbf{\Sigma}_l^{n+1}, \mathbf{\Omega}_l^{n+1})$ such that
 \begin{eqnarray}
\mathbf{\Sigma}_l^{n+1}=\frac{\tilde{\mathbf{\Sigma}}_l^{n+1}}{\lambda},\label{eq:alg4a}\\
\mathbf{\Omega}_l^{n+1}=\mathbf{I}+\sum_{k\in L\backslash \{l\}} \mathbf{H}_{lk}\mathbf{\Sigma}_k^{n+1} \mathbf{H}_{lk}^+.\label{eq:alg4b}
\end{eqnarray}
Plug $\tilde{\mathbf{\Sigma}}^{n+1}=\lambda \mathbf{\Sigma}^{n+1}$ and  $\mathbf{\Omega}=\lambda {\mathbf{\Omega}}_l^{n+1}$ into the inequality (\ref{eq:ie-max}) and combine with the inequality (\ref{eq:ie-mins}), we have 
\begin{eqnarray}
\nonumber \mathcal{F} (\mathbf{\Sigma}^n, \mathbf{\Omega}^n)\hspace{-1mm}&\leq&\hspace{-1mm} \mathcal{L}(\lambda \mathbf{\Sigma}^{n+1}, \lambda\mathbf{\Omega}^{n+1}, \mathbf{\Lambda}^n, \mb{\mu}^n) \\
\nonumber \hspace{-1mm}&=& \hspace{-1mm}\mathcal{F}( \mathbf{\Sigma}^{n+1}, \mathbf{\Omega}^{n+1})+\sum_{l\in L}(\lambda -1) \trace{\mathbf{\Lambda}_l^{n}}\\
\nonumber \hspace{-1mm}&& \hspace{-1mm}+\sum_{s\in S}\mu_s^n \left(1-\sum_{l\in L^s}\trace{\tilde{\mathbf{\Sigma}}_l^{n+1}\mathbf{Q}_l^s}\right) \\
\nonumber \hspace{-1mm}&=& \hspace{-1mm}\mathcal{F}( \mathbf{\Sigma}^{n+1}, \mathbf{\Omega}^{n+1}) +\sum_{l\in L}(\lambda -1) \trace{\mathbf{\Lambda}_l^{n}}\\
\hspace{-1mm}&\leq& \hspace{-1mm}\mathcal{F}( \mathbf{\Sigma}^{n+1}, \mathbf{\Omega}^{n+1}) ,\label{eq:mi}
 \end{eqnarray}
 where the second equality follows from equation (\ref{eq:alg3b}) and the last inequality follows from the fact that $\lambda\leq 1$.
 
\item Repeat 1)-4), we obtain a monotone increasing sequence $\{\mathcal{F}(\mathbf{\Sigma}^n, \mathbf{\Omega}^n)\}$, based on which we can conclude that $(\mathbf{\Sigma}^n, \mathbf{\Omega}^n)$ converges to a saddle point of the max-min problem (\ref{eq:maxmin-ob})-(\ref{eq:maxmin-con2}) and thus an (local) optimum of the weighted sum-rate maximization (\ref{eq:srmo})-(\ref{eq:srmc}).
\end{enumerate}

We call the above algorithm the {\em iterative minimax algorithm}; see Table \ref{tb:alg} for a formal description. 

\begin{table}[htb]
\caption{The Iterative Minimax Algorithm}
\label{tb:alg}
\begin{center}
{\ttfamily
\begin{tabular}{c}\hline\hline
1)\hspace{3mm}Initialize $\mathbf{\Sigma}_l,~l\in L $ such that~~~~~~~~~~~~~~~~~\\
~~~~~$\sum_{l\in L^s}\trace{\mathbf{\Sigma}_l \mathbf{Q}_l^s}\leq 1,~s\in S$~~~~~~~~~~~~\\
2)\hspace{3mm}$\mathbf{\Omega}_l \leftarrow \mathbf{I}+\sum_{k\in L\backslash \{l\}} \mathbf{H}_{lk}\mathbf{\Sigma}_k \mathbf{H}_{lk}^+,~l\in L$~~~~~~~~~~~~~~~~\\
3)\hspace{3mm}$\mathbf{\Lambda}_l \leftarrow w_l\left((\mathbf{\Omega}_l)^{-1}-\left(\mathbf{\Omega}_l + \mathbf{H}_{ll}\mathbf{\Sigma}_{l}\mathbf{H}_{ll}^{+}\right)^{-1}\right),~l\in L$~~~~~~~\\
4)\hspace{3mm}$\mathbf{\Phi}_l \leftarrow \sum_{s\in S^l}\mu_s \mathbf{Q}_l^s+\sum_{k\in L\backslash \{l\}} \mathbf{H}_{kl}^+\mathbf{\Lambda}_k \mathbf{H}_{kl},~l\in L$~~~~~~~\\
5)\hspace{3mm}$\tilde{\mathbf{\Sigma}}_l \leftarrow w_l\left((\mathbf{\Phi}_l)^{-1}-\left(\mathbf{\Phi}_l + \mathbf{H}_{ll}\mathbf{\Lambda}_{l}\mathbf{H}_{ll}^{+}\right)^{-1}\right),~l\in L$~~~~~~~\\
6)\hspace{3mm}$T \leftarrow \{s\in S| \sum_{l\in L^s}\trace{\mb{Q}_l^s\tilde{\mathbf{\Sigma}}_l (\mu_{\bar{s}}=0^+; \bar{s}\in S^l)}\geq 1\}$~~~~\\
7)\hspace{3mm}$\mu_s \leftarrow 0$ if $s\in S \backslash T$~~~~~~~~~~~~~~~~~~~~~~~~~~~~~~\\
8)\hspace{3mm}For $s\in T$, choose $\mu_s$ such that~~~~~~~~~~~~~~~~~\\
~~~~~~$\sum_{l\in L^s}\trace{\mb{Q}_l^s\tilde{\mathbf{\Sigma}}_l (\mu_{\bar{s}}; \bar{s}\in S^l)}= 1,~s\in T$~~~\\
9)\hspace{3mm}$\lambda \leftarrow \max_{s\in S} \sum_{l\in L^s}\trace{\tilde{\mathbf{\Sigma}}_l \mb{Q}_l^s}$~~~~~~~~~~~~~~~~~~~~~\\
10)\hspace{2mm}$\mathbf{\Sigma}_l \leftarrow \frac{\tilde{\mathbf{\Sigma}}_l}{\lambda},~l\in L$~~~~~~~~~~~~~~~~~~~~~~~~~~~~~~~~~\\
11)\hspace{2mm}Go to 2)~~~~~~~~~~~~~~~~~~~~~~~~~~~~~~~~~~~~~~\\
\hline\hline
\end{tabular}
}
\end{center}
\end{table}

\subsection{The convergence analysis}

We now study the convergence properties of the iterative minimax algorithm. The following result is immediate.
\begin{lemma}\label{thm:lp}
Under the iterative minimax algorithm, the sequence $\{C^n =\mathcal{F}(\mathbf{\Sigma}^n, \mathbf{\Omega}^n)\}$ converges to a limit point $C^*$.
\end{lemma}

\begin{proof}
Since $\mathcal{F}(\mathbf{\Sigma}, \mathbf{\Omega})$ is a continuous function and its domain (specified by the constraints (\ref{eq:maxmin-con1})-(\ref{eq:maxmin-con2})) is a compact set, $C^n$ is bounded above. By inequality (\ref{eq:mi}), the sequence $\{R^n\}$ is a monotone increasing sequence. Therefore, there exists a limit point $C^*$ such that $\lim_{n\to\infty} C^n =C^*$.
\end{proof}

\begin{theorem}
The iterative minimax algorithm converges to a saddle point $(\mb{\Sigma}^*, \mb{\Omega}^*)$ of the max-min problem (\ref{eq:maxmin-ob})-(\ref{eq:maxmin-con2}); and $\mb{\Sigma}^*$ is an optimum of the weighted sum-rate maximization (\ref{eq:srmo})-(\ref{eq:srmc}).
\end{theorem}

\begin{proof}
With Lemma \ref{thm:lp}, to show the convergence of the iterative minimax algorithm, it is enough to show that if  $\mathcal{F}(\mathbf{\Sigma}^n, \mathbf{\Omega}^n)=\mathcal{F}(\mathbf{\Sigma}^{n+1}, \mathbf{\Omega}^{n+1})$, then  $(\mathbf{\Sigma}^n, \mathbf{\Omega}^n)=(\mathbf{\Sigma}^{n+1}, \mathbf{\Omega}^{n+1})$.

From the derivation of inequality (\ref{eq:mi}), if $\mathcal{F}(\mathbf{\Sigma}^n, \mathbf{\Omega}^n)=\mathcal{F}(\mathbf{\Sigma}^{n+1}, \mathbf{\Omega}^{n+1})$, then 
\begin{eqnarray*}
\mathcal{F}(\mathbf{\Sigma}^n, \mathbf{\Omega}^n)&=&\mathcal{L}(\mathbf{\Sigma}^n, \mathbf{\Omega}^n, \mathbf{\Lambda}^n, \mathbf{\mu}^n) \\
&=& \mathcal{L}(\mathbf{\Sigma}^n, \tilde{\mathbf{\Omega}}^{n+1}, \mathbf{\Lambda}^n, \mb{\mu}^n)\\
&=&\mathcal{L}(\tilde{\mathbf{\Sigma}}^{n+1}, \tilde{\mathbf{\Omega}}^{n+1}, \mathbf{\Lambda}^n, \mu^n)\\
&=&\mathcal{L}(\mathbf{\Sigma}^{n+1}, \mathbf{\Omega}^{n+1}, \mathbf{\Lambda}^n, \mb{\mu}^n) \\
&=&\mathcal{F}(\mathbf{\Sigma}^{n+1}, \mathbf{\Omega}^{n+1}). 
\end{eqnarray*}
It follows that both $(\mathbf{\Sigma}^n, \mathbf{\Omega}^n, \mathbf{\Lambda}^n, \mathbf{\mu}^n)$ and $(\mathbf{\Sigma}^{n+1}, \mathbf{\Omega}^{n+1}, \mathbf{\Lambda}^n, \mb{\mu}^n)$ satisfy the KKT condition (the first order condition, the primal feasibility, the dual feasibility, and the complementary slackness \cite{Boyd}) of the max-min problem (\ref{eq:maxmin-ob})-(\ref{eq:maxmin-con2}), and thus both are saddle points of the max-min problem. Furthermore, for any given dual variables, the Lagrangian $\mathcal{L}$ is strictly concave in $\mathbf{\Sigma}$. So, $\mathbf{\Sigma}^n=\mathbf{\Sigma}^{n+1}$, and $\mathbf{\Omega}^n=\mathbf{\Omega}^{n+1}$ follows. Therfore, the iterative minimax algorithm converges {\em monotonically} to a saddle point of the max-min problem (\ref{eq:maxmin-ob})-(\ref{eq:maxmin-con2}). The second part of the theorem follows from the equivalence between the max-min problem  and the weighted sum-rate maximization problem.
\end{proof}

\remark{The design and convergence proof of the iterative minimax algorithm use only general convex analysis. They apply and extend to any max-min problems where the objective function is concave in the maximizing variables and convex in the minimizing variables and the constraints are convex, and thus provide a general class of algorithms for such optimization problems.}

\remark{The iterative minimax algorithm converges fairly fast and can be implemented realtime. As each link knows its own power covariance matrix and can measure/estimate its interference-plus-noise covariance matrix, the algorithm admits a distributed implementation if used as a realtime algorithm.}

\subsection{Case studies}
We now discuss a few typical cases and the corresponding iterative minimax algorithms.

\subsubsection{The network with the total power constraint}
As mentioned in Section \ref{sect:md-t}, here $|S|=1$ and $\mathbf{Q}_l=\frac{\mathbf{I}}{P_T}$, with $P_T$ the total power budget. The matrix $\tilde{\mathbf{\Sigma}}_l$ defined in Section \ref{sect:alg}  is a function of $\mu$, the dual variable associated with the total power constraint. The iterative minimax algorithm reduces to that described in Table \ref{tb:algt}.

\begin{table}[htb]
\caption{The Iterative Minimax Algorithm for the Network with the Total Power Constraint}
\label{tb:algt}
\begin{center}
{\ttfamily
\begin{tabular}{c}\hline\hline
1)\hspace{3mm}Initialize $\mathbf{\Sigma}_l,~l\in L $ such that $\sum_{l\in L}\trace{\frac{\mathbf{\Sigma}_l}{P_T}}\leq 1$~\\
2)\hspace{3mm}$\mathbf{\Omega}_l \leftarrow \mathbf{I}+\sum_{k\in L\backslash \{l\}} \mathbf{H}_{lk}\mathbf{\Sigma}_k \mathbf{H}_{lk}^+,~l\in L$~~~~~~~~~~~~~~~~\\
3)\hspace{3mm}$\mathbf{\Lambda}_l \leftarrow w_l\left((\mathbf{\Omega}_l)^{-1}-\left(\mathbf{\Omega}_l + \mathbf{H}_{ll}\mathbf{\Sigma}_{l}\mathbf{H}_{ll}^{+}\right)^{-1}\right),~l\in L$~~~~~~~\\
4)\hspace{3mm}$\mathbf{\Phi}_l \leftarrow \mu \mathbf{I}+\sum_{k\in L\backslash \{l\}} \mathbf{H}_{kl}^+\mathbf{\Lambda}_k \mathbf{H}_{kl},~l\in L$~~~~~~~~~~~~~~~\\
5)\hspace{3mm}$\tilde{\mathbf{\Sigma}}_l \leftarrow w_l\left((\mathbf{\Phi}_l)^{-1}-\left(\mathbf{\Phi}_l + \mathbf{H}_{ll}\mathbf{\Lambda}_{l}\mathbf{H}_{ll}^{+}\right)^{-1}\right),~l\in L$~~~~~~~\\
6)\hspace{3mm}$\mu \leftarrow 0$ if $\sum_{l\in L}\trace{\frac{\mb{\tilde{\mathbf{\Sigma}}}_l (0^+)}{P_T}}< 1$;otherwise choose $\mu$\\
such that $\sum_{l\in L}\trace{\frac{\tilde{\mathbf{\Sigma}}_l (\mu)}{P_T}}= 1$~~~~~~~\\
7)\hspace{3mm}$\lambda \leftarrow \sum_{l\in L}\trace{\frac{\tilde{\mathbf{\Sigma}}_l}{P_T}}$~~~~~~~~~~~~~~~~~~~~~~~~~~~~~~~\\
8)\hspace{3mm}$\mathbf{\Sigma}_l \leftarrow \frac{\tilde{\mathbf{\Sigma}}_l}{\lambda},~l\in L$~~~~~~~~~~~~~~~~~~~~~~~~~~~~~~~~~~\\
9)\hspace{3mm}Go to 2)~~~~~~~~~~~~~~~~~~~~~~~~~~~~~~~~~~~~~~~\\
\hline\hline
\end{tabular}
}
\end{center}
\end{table}

The above algorithm is different from the algorithm proposed in the previous work \cite{Li-2014-Infocom}. The algorithm in \cite{Li-2014-Infocom} uses the fact that the total power constraint is tight at an optimum, and normalizes $\mu$ such that $\sum_{l\in L}\trace{\frac{\mu}{P_T}\tilde{\mathbf{\Sigma}}_l}=1$, i.e., the algorithm enforces the tightness of the total power constraint at the initial point and each iteration. In contrast, our algorithm enforces the complementary slackness condition at each iteration and can start with any feasible $\mathbf{\Sigma}$.

\subsubsection{The network with the per-link power constraints}
Here the set $S=L$ and  $\mathbf{Q}_l=\frac{\mathbf{I}}{P_l}$, with $P_l$ the power budget at each link $l\in L$. The matrix $\tilde{\mathbf{\Sigma}}_l$ defined in Section \ref{sect:alg}  is a function of $\mu_l$, the dual variable associated with the power constraint at link $l$. The iterative minimax algorithm reduces to that described in Table \ref{tb:algp}.

\begin{table}[htb]
\caption{The Iterative Minimax Algorithm for the Network with the Per-Link Power Constraints}
\label{tb:algp}
\begin{center}
{\ttfamily
\begin{tabular}{c}\hline\hline
1)\hspace{3mm}Initialize $\mathbf{\Sigma}_l,~l\in L $ such that $\trace{\frac{\mathbf{\Sigma}_l}{P_l}}\leq 1$~~~~~~\\
2)\hspace{3mm}$\mathbf{\Omega}_l \leftarrow \mathbf{I}+\sum_{k\in L\backslash \{l\}} \mathbf{H}_{lk}\mathbf{\Sigma}_k \mathbf{H}_{lk}^+,~l\in L$~~~~~~~~~~~~~~~~\\
3)\hspace{3mm}$\mathbf{\Lambda}_l \leftarrow w_l\left((\mathbf{\Omega}_l)^{-1}-\left(\mathbf{\Omega}_l + \mathbf{H}_{ll}\mathbf{\Sigma}_{l}\mathbf{H}_{ll}^{+}\right)^{-1}\right),~l\in L$~~~~~~~\\
4)\hspace{3mm}$\mathbf{\Phi}_l \leftarrow \mu \mathbf{I}+\sum_{k\in L\backslash \{l\}} \mathbf{H}_{kl}^+\mathbf{\Lambda}_k \mathbf{H}_{kl},~l\in L$~~~~~~~~~~~~~~~\\
5)\hspace{3mm}$\tilde{\mathbf{\Sigma}}_l \leftarrow w_l\left((\mathbf{\Phi}_l)^{-1}-\left(\mathbf{\Phi}_l + \mathbf{H}_{ll}\mathbf{\Lambda}_{l}\mathbf{H}_{ll}^{+}\right)^{-1}\right),~l\in L$~~~~~~~\\
6)\hspace{3mm}$\mu \leftarrow 0$ if $\trace{\frac{\mb{\tilde{\mathbf{\Sigma}}}_l (0^+)}{P_l}}< 1$; otherwise, choose $\mu$~~~\\
such that $\trace{\frac{\tilde{\mathbf{\Sigma}}_l (\mu)}{P_l}}= 1$~~~~~~~~~~~\\
7)\hspace{3mm}$\lambda \leftarrow \max_{l\in L}\trace{\frac{\tilde{\mathbf{\Sigma}}_l}{P_T}}$~~~~~~~~~~~~~~~~~~~~~~~~~~~~~\\
8)\hspace{3mm}$\mathbf{\Sigma}_l \leftarrow \frac{\tilde{\mathbf{\Sigma}}_l}{\lambda},~l\in L$~~~~~~~~~~~~~~~~~~~~~~~~~~~~~~~~~\\
9)\hspace{3mm}Go to 2)~~~~~~~~~~~~~~~~~~~~~~~~~~~~~~~~~~~~~~\\
\hline\hline
\end{tabular}
}
\end{center}
\end{table}

\section{Numerical Examples}\label{sec:sim}
In this section, we provide numerical examples to complement the analysis in the previous sections. Consider a network with $L = 10$ links,  corresponding to $10$  transmitter-receiver pairs that interfere with each other. Each link is equipped with 3 (4) antennas at its transmitter (receiver). The channel matrices have zero-mean, unit-variance, i.i.d. complex Gaussian entries. We will consider and compare the networks with low, moderate, and high interference, which are characterized by scaling the interference channel matrices $\mathbf{H}_{ij},~ i \neq j$ with a factor of $0.1$, $1$, and $5$ respectively. The weights $w_l$'s are uniformly drawn from $[0.5, 1]$, for the case with total power constraint $P_T=10$, and for the case with the per-link power constraints $P_l$'s are uniformly drawn from $\{1,2,\cdots,10\}$.

For the computation, we use SDPT3 \cite{toh1999sdpt3} combined with the problem parser YALMIP \cite{lofberg2004yalmip}.
The algorithm implementation is straightforward except for finding $\mu$, for which we use a bisection search method.

\paragraph{The network with the total power constraint}
Figures \ref{fig:lowt}, \ref{fig:modt} and \ref{fig:hight} show the monotonic convergence of our algorithm in a network with the total power constraint. We see that the convergence speed depends on the strength of interference. As the interference becomes stronger, the weighted sum-rate becomes highly non-convex.  This intrinsic difficulty of the problem makes the convergence slow.
However, in the network with  low and moderate interference, the algorithm shows very fast convergence.
Also note that the stronger the interference, the smaller the weighted sum-rate is. 

\begin{figure}[ht!]
\centering
\includegraphics[width=0.4\textwidth]{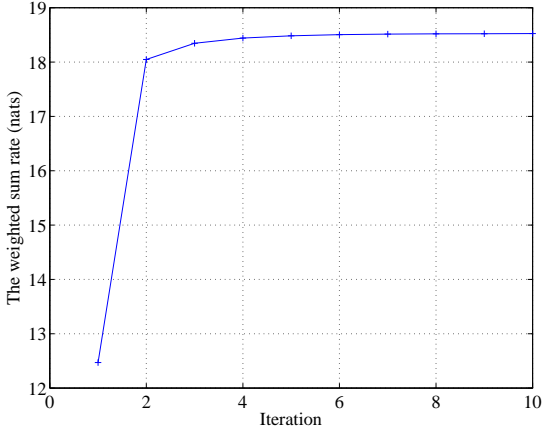}
\caption{The network with low interference and total power constraint.}
\label{fig:lowt}
\end{figure}

\begin{figure}[ht!]
\centering
\includegraphics[width=0.4\textwidth]{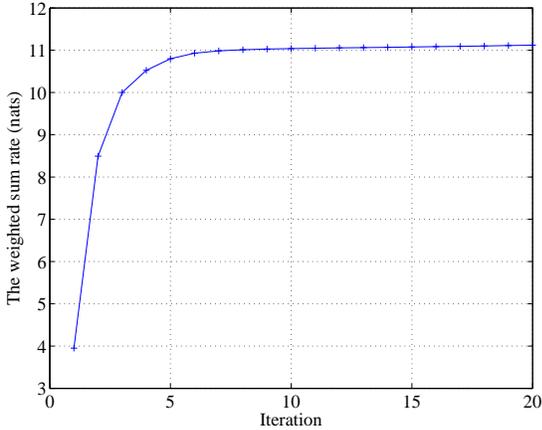}
\caption{The network with moderate interference and total power constraint.}
\label{fig:modt}
\end{figure}

\begin{figure}[ht!]
\centering
\includegraphics[width=0.4\textwidth]{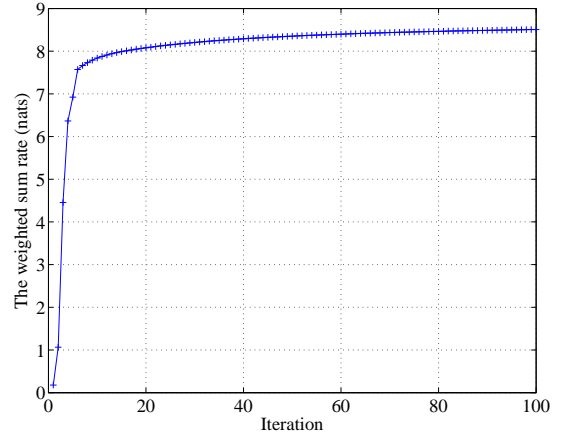}
\caption{The network with high interference and total power constraint.}
\label{fig:hight}
\end{figure}

\paragraph{The network with the per-link power constraints}
Figures \ref{fig:low}, \ref{fig:mod} and \ref{fig:high} show the monotonic convergence of our algorithm in a network with the per-link power constraints. Again, we see that the stronger the interference, the slower the algorithm converges; but in the network with  low and moderate interference, the algorithm shows fast convergence.

\begin{figure}[ht!]
\centering
\includegraphics[width=0.4\textwidth]{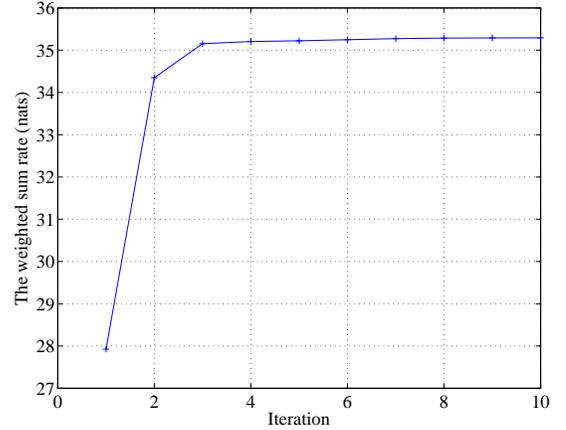}
\caption{The network with low interference and per-link power constraints.}
\label{fig:low}
\end{figure}

\begin{figure}[ht!]
\centering
\includegraphics[width=0.4\textwidth]{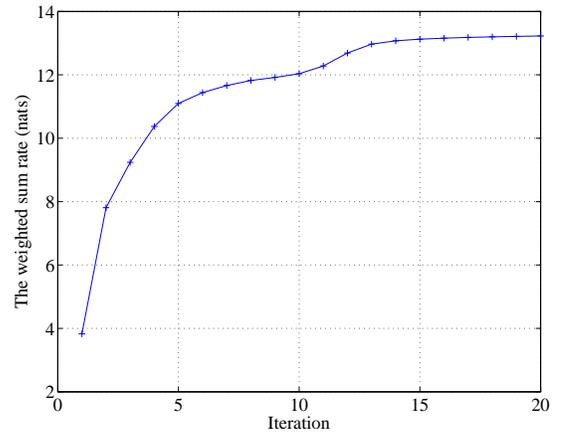}
\caption{The network with moderate interference and per-link power constraints.}
\label{fig:mod}
\end{figure}

\begin{figure}[ht!]
\centering
\includegraphics[width=0.4\textwidth]{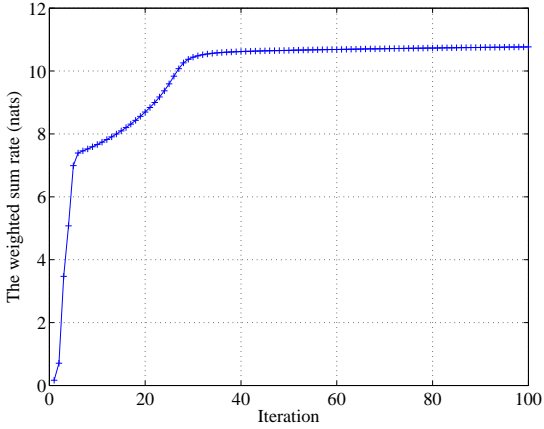}
\caption{The network with high interference and per-link power constraints.}
\label{fig:high}
\end{figure}

\paragraph{Complexity Analysis}
We have evaluated in the above the monotonic convergence of our algorithm in terms of the number of iterations. We now analyze the complexity of each iteration. Recall that $L$ is the number of data links, and for simplicity, assume that each link has $N$ transmit (and receive) antennas, so the resulting $\mathbf{\Sigma}_{l}$ is an $N\times N$ matrix.
Suppose that we use the straightforward matrix multiplication and inversion, 
then the complexity of these operations are $O(N^{3})$.
In each iteration, $\mathbf{\Omega}_{l}$ incurs a complexity of $O(LN^{3})$, and so does $\mathbf{\Omega}_{l}+\mathbf{H}_{l,l}\mathbf{\Sigma}_{l}^{(n+1)}\mathbf{H}_{l,l}^{+}$.
Furthermore, $\mathbf{\Phi}_l$ incurs a complexity of $O(LN^{3})$, and so do $\tilde{\mathbf{\Sigma}}_l$ and 
$\mathbf{\Sigma}_l$.  Since we need $L$ of these operations, the total complexity is $O(L^2N^{3})$. If we use faster matrix multiplication such as the one in \cite{williams2012multiplying} that has a complexity of $O(N^{2.3727})$, we can reduce computational complexity at each iteration to $O(L^2N^{2.3727})$.

%\section{Comments and further work}
%The minimax duality perspective allows natural extensions to the weighted sum-rate maximization problem with general constraints, and suggests new algorithms for other problems that entail the minimax duality. 
%\begin{itemize}
%\item Extend to more general settings, e.g., with a separate power constraint at each transmitter. In a general setting, certain power constraints may not be tight. How to modify the algorithm to accommodate situations like this?
%\item Explore insights that may lead to other algorithms.
%\item Similar algorithms can be used to solve those problems that entail the minimax duality. This potentially lead to algorithms for a large class of problems.
%\item Explore a duality perspective on other types of problems.
%\item Connect it to a game-theoretic interpretation.
%\item to be continued...
%\end{itemize}

\section{Conclusion}
We take a new perspective on the weighted sum-rate maximization in the MIMO interference network, by formulating an equivalent max-min problem. The Lagrangian duality of the equivalent max-min problem provides an elegant way to establish the sum-rate duality between an interference network and its reciprocal when such a duality exists, and more importantly, suggests a novel iterative minimax algorithm for the weighted sum-rate maximization. The design and convergence proof of the iterative minimax algorithm use only general convex analysis and matrix analysis. They apply and extend to any max-min problems where the objective function is concave in the maximizing variables and convex in the minimizing variables and the constraints are convex, and thus provides a general class of algorithms for such optimization problems. This paper presents a promising step and lends hope for establishing a general framework based on the minimax Lagrangian duality for characterizing the weighted sum-rate and developing efficient algorithms for general MIMO interference networks. 

\section*{Appendix: Proof of Lemma \ref{thm:es}}
Before we present the proof, we first define an extended difference of logdet function. Let $\mb{A}, \mb{B}\in \mathcal{S}^n_+$, the difference of logdet function  
$$F(\mb{A}, \mb{B}) = \log\left|\mb{A}+\mb{B}\right|-\log\left|\mb{B}\right|$$  
is not well-defined if $\mb{B}$ is not positive definite. If there exists a nonsingular square matrix $\mb{T}$ such that 
\begin{equation*}
\mb{T}^+\mb{A}\mb{T}=\left[\begin{array}{cc}
\mb{A}_1 &  \\
 & \mb{0} 
 \end{array}
\right],\ \  \mb{T}^+\mb{B}\mb{T}=\left[\begin{array}{cc}
\mb{B}_1 &  \\
 & \mb{0} 
 \end{array}
\right]
\end{equation*}
where $\mb{A}_1\in \mathcal{S}^m_+$, $\mb{B}_1\in \mathcal{S}^m_{++}$ for some $m\leq n$, then we can define an extended difference of logdet function:
\begin{equation*}
F(\mb{A}, \mb{B}) := \log\left|\mb{A}_1+\mb{B}_1\right|-\log\left|\mb{B}_1\right|.
\end{equation*}
With the definition of the above extended function,  matrix inverse resulting from the derivative of logdet function is pseudo inverse when the matrix involved is singular. In the following, {\em a difference of logdet function is meant to be the extended difference of logdet function, and matrix inverse is pseudo inverse when the matrix involved is singular}. 

We now come to the proof of Lemma \ref{thm:es}. For simplicity of presentation and without loss of generality, we reload notations and consider the following problem:
\begin{equation}\label{eq:ss}
\max_{\mathbf{\Sigma}\succeq 0} \min_{\mathbf{\Omega}\succeq 0}\ \ 
\log|\mb{\Omega} + \mb{H} \mb{\Sigma} \mb{H}^{+}| - \log|\mb{\Omega}| + \trace{\mb{\Lambda} \mb{\Omega}} - \trace{\mb{\Phi} \mb{\Sigma}}
\end{equation}
where $\mb{\Lambda} \succeq 0$ and $\mb{\Phi} \succeq 0$. The key idea of the proof is to show that problem \eqref{eq:ss} is equivalent to a problem with $\mb{\Omega}$ restricted to 
$\mb{\Omega} = \mb{H} \mb{X} \mb{H}^{+},~\mb{X}\succeq 0$.

\begin{lemma}\label{thm:eq}
The problem \eqref{eq:ss} is equivalent to the following problem:
\begin{eqnarray}
\nonumber \max_{\mathbf{\Sigma}\succeq 0}\min_{\mathbf{\Omega}\succeq 0}
\hspace{-4mm}&& \log|\mb{\Omega} + \mb{H} \mb{\Sigma} \mb{H}^{+}| - \log|\mb{\Omega}| + \trace{\mb{\Lambda} \mb{\Omega}} - \trace{\mb{\Phi} \mb{\Sigma}}\\
\label{eq:sse-ob}\\
\text{s.t.} \hspace{-4mm}&& \mb{\Omega} = \mb{H} \mb{X} \mb{H}^{+},\ \ \mb{X}\succeq 0.\label{eq:sse-con}
\end{eqnarray}
\end{lemma}

\begin{proof}
Since $\mb{H} \mb{\Sigma} \mb{H}^{+} \succeq 0$ and $\mb{\Lambda} \succeq 0$, there exists a nonsingular square matrix $\mb{T}$ such that
\BEAS
\mb{T}\mb{\Lambda}\mb{T}^{+} &=& \left[\begin{array}{cccc}
\mb{S}_1 & &  & \\
 & \mb{0} &  & \\
 & & \mb{S}_3 & \\
 &  &  & \mb{0}
\end{array}
\right],\\
(\mb{T}^{+})^{-1}\mb{H} \mb{\Sigma} \mb{H}^{+}\mb{T}^{-1}&=& \left[\begin{array}{cccc}
\mb{S}_1 &  & & \\
 & \mb{S}_2 &  & \\
 &  & \mb{0}& \\
& &  & \mb{0}
\end{array}
\right],
\EEAS
where $\mb{S}_1, \mb{S}_2, \mb{S}_3$  are diagonal and positive definite; see, e.g., Theorem 3.22 in~\cite{Zhou95}. Let  $\mb{\Omega} = \mb{T}^{+}\tilde{\mb{\Omega}} \mb{T}$, problem \eqref{eq:ss} becomes
\begin{eqnarray}
\nonumber \max_{\mathbf{\Sigma}\succeq 0} \min_{\tilde{\mathbf{\Omega}}\succeq 0}\hspace{-4mm}&&
 \log|\tilde{\mb{\Omega}} + (\mb{T}^{+})^{-1}\mb{H} \mb{\Sigma} \mb{H}^{+}\mb{T}^{-1}| - \log|\tilde{\mb{\Omega}}| \\
\nonumber \hspace{-4mm}&&+ \trace{\mb{T}\mb{\Lambda}\mb{T}^{+} \tilde{\mb{\Omega}}} - \trace{\mb{\Phi} \mb{\Sigma}}.
\end{eqnarray}

Now, consider those terms in the objective function that depend on $\tilde{\mathbf{\Omega}}$:
\BEAS
\hspace{-2mm}&&\hspace{-1mm} \tilde{\mathcal{L}}(\tilde{\mathbf{\Omega}})\\
\hspace{-2mm}&=&\hspace{-1mm}\log|\tilde{\mb{\Omega}} + (\mb{T}^{+})^{-1}\mb{H} \mb{\Sigma} \mb{H}^{+}\mb{T}^{-1}| - \log|\tilde{\mb{\Omega}}| + \trace{\mb{T}\mb{\Lambda}\mb{T}^{+} \tilde{\mb{\Omega}}}\\
\hspace{-2mm}&=& \hspace{-1mm}\log\left|
\left[
\begin{array}{cccc}
\tilde{\mb{\Omega}}_{11} + \mb{S}_1 & \tilde{\mb{\Omega}}_{12} & \tilde{\mb{\Omega}}_{13} & \tilde{\mb{\Omega}}_{14}\\
\tilde{\mb{\Omega}}_{12}^+ & \tilde{\mb{\Omega}}_{22} + \mb{S}_2 & \tilde{\mb{\Omega}}_{23} & \tilde{\mb{\Omega}}_{24}\\
\tilde{\mb{\Omega}}_{13}^+ & \tilde{\mb{\Omega}}_{23}^+ &   \tilde{\mb{\Omega}}_{33} & \tilde{\mb{\Omega}}_{34}\\
\tilde{\mb{\Omega}}_{14}^+ & \tilde{\mb{\Omega}}_{24}^+ & \tilde{\mb{\Omega}}_{34}^+ &  \tilde{\mb{\Omega}}_{44}
\end{array}
\right]\right|\\
\hspace{-2mm}&&\hspace{-1mm}
-
\log\left|
\left[
\begin{array}{cccc}
\tilde{\mb{\Omega}}_{11}  & \tilde{\mb{\Omega}}_{12} & \tilde{\mb{\Omega}}_{13} & \tilde{\mb{\Omega}}_{14}\\
\tilde{\mb{\Omega}}_{12}^+  & \tilde{\mb{\Omega}}_{22} & \tilde{\mb{\Omega}}_{23} & \tilde{\mb{\Omega}}_{24}\\
\tilde{\mb{\Omega}}_{13}^+ & \tilde{\mb{\Omega}}_{23}^+&   \tilde{\mb{\Omega}}_{33} & \tilde{\mb{\Omega}}_{34}\\
\tilde{\mb{\Omega}}_{14}^+ & \tilde{\mb{\Omega}}_{24}^+ & \tilde{\mb{\Omega}}_{34}^+ &  \tilde{\mb{\Omega}}_{44}
\end{array}
\right]\right|\\
\hspace{-2mm}&&\hspace{-1mm}
+
\trace{\mb{S}_1 \tilde{\mb{\Omega}}_{11}} + \trace{\mb{S}_3 \tilde{\mb{\Omega}}_{33}}
\EEAS
and its minimization over $\tilde{\mathbf{\Omega}}\succeq 0$. By the determinant formula for block matrix, when $A$  is invertible
$\left|\left[\begin{array}{cc}
A & B\\
C & D
\end{array}\right]\right|=\left|A\right| \left|D-CA^{-1}B\right|$, and the fact that the determinant is a continuous function, we have
\BEAS
\tilde{\mathcal{L}}(\tilde{\mathbf{\Omega}})
&\geq&
\log\left|
\left[
\begin{array}{ccc}
\tilde{\mb{\Omega}}_{11} + \mb{S}_1 & \tilde{\mb{\Omega}}_{12} & \tilde{\mb{\Omega}}_{13}\\
\tilde{\mb{\Omega}}_{12}^+ & \tilde{\mb{\Omega}}_{22} + \mb{S}_2 & \tilde{\mb{\Omega}}_{23}\\
\tilde{\mb{\Omega}}_{13}^+ & \tilde{\mb{\Omega}}_{23}^+ &   \tilde{\mb{\Omega}}_{33}
\end{array}
\right]\right|\\
&&-
\log\left|
\left[
\begin{array}{ccc}
\tilde{\mb{\Omega}}_{11}  & \tilde{\mb{\Omega}}_{12} & \tilde{\mb{\Omega}}_{13}\\
\tilde{\mb{\Omega}}_{12}^+ & \tilde{\mb{\Omega}}_{22} & \tilde{\mb{\Omega}}_{23}\\
\tilde{\mb{\Omega}}_{13}^+ & \tilde{\mb{\Omega}}_{23}^+ &   \tilde{\mb{\Omega}}_{33}
\end{array}
\right]\right|\\
&&
+
\trace{\mb{S}_1 \tilde{\mb{\Omega}}_{11}} + \trace{\mb{S}_3 \tilde{\mb{\Omega}}_{33}},
\EEAS
where the equality holds when $\tilde{\mb{\Omega}}_{44} \rightarrow \infty$ but can be achieved when $\tilde{\mb{\Omega}}_{i4} = \mb{0}$ for all $i=1, 2, 3, 4$.\footnote{By the determinant formula and the Sylvester's criterion for positive semidefinite matrix, if  $\tilde{\mb{\Omega}}_{i4} = {0}$, then $\tilde{\mb{\Omega}}_{i4} = \mb{0}$ for all $i=1, 2, 3$.} We will restrict $\tilde{\mathbf{\Omega}}$ to those with $\tilde{\mb{\Omega}}_{i4} = \mb{0}$ for all $i=1, 2, 3, 4$, as the equality is achieved at one of those matrices. 

Since $\mb{S}_3 \succeq 0$ and $\tilde{\mb{\Omega}}_{33} \succeq {0}$, $\trace{\mb{S}_3 \tilde{\mb{\Omega}}_{33}} \geq 0$. We further have\BEAS
\tilde{\mathcal{L}}(\tilde{\mathbf{\Omega}})
&\geq&
\log\left|
\left[
\begin{array}{cc}
\tilde{\mb{\Omega}}_{11} + \mb{S}_1 & \tilde{\mb{\Omega}}_{12}\\
\tilde{\mb{\Omega}}_{12}^+ & \tilde{\mb{\Omega}}_{22} + \mb{S}_2
\end{array}
\right]\right|\\
&&-
\log\left|
\left[
\begin{array}{cc}
\tilde{\mb{\Omega}}_{11}  & \tilde{\mb{\Omega}}_{12}\\
\tilde{\mb{\Omega}}_{12}^+ & \tilde{\mb{\Omega}}_{22}
\end{array}
\right]\right|
+
\trace{\mb{S}_1 \tilde{\mb{\Omega}}_{11}},
\EEAS
where the equality is achieved when additionally $\tilde{\mb{\Omega}}_{i3} = \mb{0}$ for all $i=1, 2, 3$. Therefore, we conclude that there exists a minimizer $\tilde{\mb{\Omega}}^{\star}$ with the form:
\BEAS
\tilde{\mb{\Omega}}^{\star}
&=&
\left[
\begin{array}{cccc}
\tilde{\mb{\Omega}}_{11}  & \tilde{\mb{\Omega}}_{12} &  & \\
\tilde{\mb{\Omega}}_{12}^+& \tilde{\mb{\Omega}}_{22} &  & \\
 &  &   \mb{0} & \\
&  &  &  \mb{0}
\end{array}
\right].
\EEAS

The above manipulation is to restrict the problem to an equivalent, {\em truncated system} where we ignore the interference-plus-noise of a channel whose signal is zero. As mentioned in Section \ref{sect:mdd}, intuitively, the equivalence of this truncated system to the original max-min problem follows from the fact that when the signal is zero it does not matter what the interference-plus-noise is.

Now, consider a vector $v$ such that $\mb{H}^{+} v = 0$, we have
\BEAS
v^{+}\mb{H}\mb{\Sigma}\mb{H}^Tv &=& v^{+}\mb{T}^{+} \left[\begin{array}{cccc}
\mb{S}_1 &  &  & \\
 & \mb{S}_2 &  & \\
&  & \mb{0}& \\
 &  &  & \mb{0}
\end{array}
\right]
\mb{T}
v\\
&=&0,						
\EEAS
which implies
\BEAS
\mb{T}v &=& 
\left[
\begin{array}{cccc}
\mb{0}&
\mb{0}&
\bar{\mb{v}}_3&
\bar{\mb{v}}_4
\end{array}
\right]^T.
\EEAS
Therefore,
\BEAS
v^{+}\mb{\Omega}^{\star} v 
&=& v^{+} \mb{T}^{+} \tilde{\mb{\Omega}}^{\star} \mb{T} v\\
&=&
v^{+}\mb{T}^{+}
\left[
\begin{array}{cccc}
\tilde{\mb{\Omega}}_{11}  & \tilde{\mb{\Omega}}_{12} &  & \\
\tilde{\mb{\Omega}}_{12}^+ & \tilde{\mb{\Omega}}_{22} & & \\
&  &   \mb{0} & \\
&  &  &  \mb{0}
\end{array}
\right]
\mb{T}
v\\
&=& {0}.
\EEAS
This implies the null space $\mathcal{N}(H^{+}) \subset \mathcal{N}(\mb{\Omega}^{\star})$, and further, the range
$\mathcal{R}(H) \supset \mathcal{R}({\mb{\Omega}^{\star}}^{+}) = \mathcal{R}({\mb{\Omega}^{\star}})$. Therefore, there exists a matrix $\mb{X}\succeq 0$ such that
\BEAS
\mb{\Omega}^{\star} &=& \mb{H}\mb{X} \mb{H}^{+}.
\EEAS
We conclude that there exists an optimal solution with $\mb{\Omega}^{\star} = \mb{H}\mb{X} \mb{H}^{+}$, and thus problem \eqref{eq:ss} and problem \eqref{eq:sse-ob}-\eqref{eq:sse-con} are equivalent.
\end{proof}

With Lemma \ref{thm:eq}, we are ready to present the explicit saddle point solution. Consider the logdet terms in the objective function:
\BEAS
&&\log|\mb{\Omega} + \mb{H} \mb{\Sigma} \mb{H}^{+}| - \log|\mb{\Omega}|\\
&=& \log|\mb{H}\mb{X}\mb{H}^{+} + \mb{H} \mb{\Sigma} \mb{H}^{+}| - \log|\mb{H}\mb{X}\mb{H}^{+}|\\
&=& \log|\mb{H}^{+}\mb{H}(\mb{X}+\mb{\Sigma})|-\log|\mb{H}^{+}\mb{H}\mb{X}|\\
&=& \log|\mb{X}+\mb{\Sigma}| - \log|\mb{X}|.
\EEAS
The singularity issue comes out when $\mb{H}^{+}\mb{H}$ is not invertible, but this can be handled by adding a small term $\kappa \mb{I},~\kappa>0$ to $\mb{H}^{+}\mb{H}$ and then taking the limit $\kappa\to 0$. Thus, we can transform problem \eqref{eq:ss} into the following simple one:
\begin{equation}\label{eq:sp}
 \max_{\mathbf{\Sigma}\succeq 0}\min_{\mathbf{X}\succeq 0}\ \
 \log|\mb{X} + \mb{\Sigma}| - \log|\mb{X}| + \trace{\mb{H}^{+}\mb{\Lambda}\mb{H} \mb{X}} - \trace{\mb{\Phi} \mb{\Sigma}}.
 \end{equation}
 By the first order optimality condition for the saddle point, we have
 \BEAS
 (\mb{X} + \mb{\Sigma})^{-1} - \mb{\Phi} &=& 0,\\
(\mb{X} + \mb{\Sigma})^{-1} - \mb{X}^{-1} +\mb{H}^{+}\mb{\Lambda}\mb{H} &=& 0,
\EEAS
from which we obtain the following explicit saddle point solution:
\begin{eqnarray}
\mb{\Sigma} &=& \mb{\Phi}^{-1} - (\mb{\Phi} + \mb{H}^{+}\mb{\Lambda}\mb{H})^{-1},\label{eq:ses-1}\\
\nonumber \mb{X} &=& (\mb{\Phi} + \mb{H}^{+}\mb{\Lambda}\mb{H})^{-1}, 
\end{eqnarray}
and in terms of $\Omega$ we have
\begin{eqnarray}
 \mb{\Omega} &=& \mb{H}\mb{X}\mb{H}^{+}
				=\mb{H}(\mb{\Phi} + \mb{H}^{+}\mb{\Lambda}\mb{H})^{-1}\mb{H}^{+}.\label{eq:ses-2}
\end{eqnarray}

Note that problem \eqref{eq:sp} is well-defined only when $\mb{\Sigma},~\mb{X}$ satisfy the property specified for matrices $\mb{A},~\mb{B}$ in the beginning of this Appendix. This is verified as follows.
\begin{proposition}
The objective function in problem \eqref{eq:sp} is well-defined for matrices $\mb{\Sigma},~\mb{X}$ that are given by \eqref{eq:ses-1}-\eqref{eq:ses-2}.
\end{proposition}
\begin{proof}
Let $\mb{\Psi} = \mb{\Phi} + \mb{H}^+ \mb{\Lambda} \mb{H}$, we have that the null space $\mathcal{N}(\mb{\Psi}) \subset \mathcal{N}(\mb{\Phi})$. To see this, note that $\mb{\Phi} \succeq 0$ and $\mb{H}^+ \mb{\Lambda} \mb{H} \succeq 0$. Suppose ${v} \in \mathcal{N}(\mb{\Psi})$, then
${v}^T\mb{\Psi}{v} = {v}^T\mb{\Phi} {v} + {v}^T \mb{H}^+ \mb{\Lambda} \mb{H}{v}= 0$. As each term is nonnegative, ${v}^T\mb{\Phi} {v} = 0$, i.e., ${v} \in \mathcal{N}(\mb{\Phi})$.

Since $\mathcal{N}(\mb{\Psi}) \subset \mathcal{N}(\mb{\Phi})$, there exists a unitary matrix $\mb{U}$ such that
\BEAS
\mb{U}^{+} \mb{\Psi} \mb{U} = \left[
\begin{array}{cc}
\mb{\Psi}_1 & \\
 & \mb{0}
\end{array}
\right],\ \ \mb{U}^+ \mb{\Phi} \mb{U} = \left[
\begin{array}{cc}
\mb{\Phi}_1 & \\
& \mb{0}
\end{array}
\right]
\EEAS
where $\mb{\Psi}_1\succ 0$ and $\mb{\Phi}_1\succeq 0$. By equations \eqref{eq:ses-1}-\eqref{eq:ses-2},
\BEAS
\mb{X}  = \mb{U}\left[
\begin{array}{cc}
\mb{\Psi}_1^{-1} & \\
& \mb{0}
\end{array}
\right] \mb{U}^+,\ \mb{\Sigma} = \mb{U}\left[
\begin{array}{cc}
\mb{\Phi}_1^{-1} - \mb{\Psi}_1^{-1} & \\
& \mb{0}
\end{array}
\right]\mb{U}^+.
\EEAS
Note that $\mb{\Psi}_1^{-1}\succ 0$, so by the definition of the extended difference of logdet function, the objective function in problem \eqref{eq:sp} is well-defined. %for matrices $\mb{\Sigma},~\mb{X}$ that are given by \eqref{eq:ses-1}-\eqref{eq:ses-2}.
\end{proof}

With \eqref{eq:ses-1}-\eqref{eq:ses-2}, we can easily recover the explicit solution \eqref{eq:es1}-\eqref{eq:es2}. This concludes the proof of Lemma \ref{thm:es}.

%\begin{thebibliography}{1}
\bibliographystyle{plain}
\bibliography{ChenReferences,cvxccv}

%\bibitem{Li-2013-Infocom}
%X. Li, S. You, L. Chen, A. Liu and Y. Liu,
%\newblock{A New Algorithm for the Weighted Sum Rate Maximization in MIMO Interference Networks},
%submited, {\em IEEE Infocom}, 2014.
%\bibitem{Yu-2006-IT}
%W. Yu,
%\newblock{Uplink-Downlink Duality Via Minimax Duality}, {\em IEEE Transactions on Information Theory}, 52(2): 361-374, 2014.
%\bibitem{BT89}
%D. P. Bertsekas and J. N. Tsitsiklis,
%\newblock {\em Parallel and Distributed Computation},
%Prentice Hall, 1989.
%\bibitem{Fud91}
%D. Fudenburg and J. Tirole,
%\newblock {\em Game Theory}, The MIT Press, 1991.
%\end{thebibliography}

\end{document}